\documentclass[sigconf]{aamas} 

\pdfoutput=1

\usepackage{balance} %

\setcopyright{ifaamas}
\acmConference[AAMAS '24]{Proc.\@ of the 23rd International Conference
on Autonomous Agents and Multiagent Systems (AAMAS 2024)}{May 6 -- 10, 2024}
{Auckland, New Zealand}{N.~Alechina, V.~Dignum, M.~Dastani, J.S.~Sichman (eds.)}
\copyrightyear{2024}
\acmYear{2024}
\acmDOI{}
\acmPrice{}
\acmISBN{}

\acmSubmissionID{1057}

\title[Monitoring Second-Order Hyperproperties]{Monitoring Second-Order Hyperproperties}

\author{Raven Beutner}
\affiliation{
  \institution{CISPA Helmholtz Center for \\Information Security}
  \country{Germany}}
\orcid{0000-0001-6234-5651}

\author{Bernd Finkbeiner}
\affiliation{
  \institution{CISPA Helmholtz Center for \\Information Security}
  \country{Germany}}
\orcid{0000-0002-4280-8441}

\author{Hadar Frenkel}
\affiliation{
  \institution{CISPA Helmholtz Center for \\Information Security}
  \country{Germany}}
\orcid{0000-0002-3566-0338}

\author{Niklas Metzger}
\affiliation{
  \institution{CISPA Helmholtz Center for \\Information Security}
  \country{Germany}}
\orcid{0000-0003-3184-6335}

\begin{abstract}
Hyperproperties express the relationship between multiple executions of a system. 
This is needed in many AI-related fields, such as knowledge representation and planning, to capture system properties related to knowledge, information flow, and privacy. In this paper, we study the monitoring of complex hyperproperties at runtime. Previous work in this area has either focused on the simpler problem of monitoring trace properties (which are sets of traces, while hyperproperties are sets of sets of traces) or on monitoring first-order hyperproperties, which are expressible in temporal logics with first-order quantification over traces, such as HyperLTL. We present the first monitoring algorithm for the much more expressive class of second-order hyperproperties. Second-order hyperproperties include  system properties like common knowledge, which cannot be expressed in first-order logics like HyperLTL.

We introduce \sohyperltlstar{}, a temporal logic over finite traces that allows for second-order quantification over sets of traces.
We study the monitoring problem in two fundamental execution models: (1) the parallel model, where a fixed number of traces is monitored in parallel, and (2) the sequential model, where an unbounded number of traces is observed sequentially, one trace after the other.
For the parallel model, we show that the monitoring of the second-order hyperproperties of \sohyperltlstar{} can be reduced to monitoring first-order hyperproperties. 
For the sequential model, we present a monitoring algorithm that handles second-order quantification efficiently, exploiting optimizations based on the monotonicity of subformulas, graph-based storing of executions, and fixpoint hashing.
We present experimental results from a range of benchmarks, including examples from common knowledge and planning. 
\end{abstract}

\keywords{Runtime Verification,
Common Knowledge,
Multi-agent Systems
}

\newcommand{\BibTeX}{\rm B\kern-.05em{\sc i\kern-.025em b}\kern-.08em\TeX}

\usepackage{ltl}
\usepackage{mathtools}

\usepackage{array}

\usepackage{algorithm}
\usepackage{algorithmic}

\usepackage{multirow}

\usepackage{enumitem}
\usepackage{amssymb}
\usepackage{amsmath}
\usepackage{graphicx}
\usepackage{amsfonts} %
\usepackage{stmaryrd} %
\usepackage{xcolor}
\usepackage{booktabs}
\usepackage{amsthm}
\usepackage{cleveref}
\usepackage{mathpartir}
\usepackage{tikz}

\newtheorem{definition}{Definition}
\newtheorem{theorem}{Theorem}
\newtheorem{proposition}{Proposition}
\newtheorem{example}{Example}

\usepackage{tablefootnote}
\usepackage{makecell}
\usepackage{todonotes}
\usepackage{comment}
\usepackage{subcaption}
\usepackage{listings}
\usepackage{relsize}

\usepackage{caption}

\usepackage{multirow}

\allowdisplaybreaks

\newcommand{\tracevars}{\mathcal{V}}
\newcommand{\sovars}{\mathcal{W}}
\newcommand{\systemvar}{\mathit{sys}}

\renewcommand{\models}{\vDash}

\newcommand{\ap}{\text{AP}}

\newcommand{\U}{\LTLuntil}
\newcommand{\X}{\LTLnext}

\newcommand{\quant}{\mathbb{Q}}

\newcommand{\nat}{\mathbb{N}}

\newcommand{\traceSet}{\mathbb{T}}

\newcommand{\ldot}{\mathpunct{.}}

\newcommand{\sohyperltl}[0]{\text{\normalfont{Hyper\textsuperscript{2}LTL}}}%
\newcommand{\sohyperltlstar}[0]{$\text{Hyper}^2\text{LTL}_f$}
\newcommand{\ltlf}[0]{$\text{LTL}_f$}

\newcommand{\titlesohyperltlstar}[0]{$\text{Hyper}^2\text{LTL}_f$}
\newcommand{\ltl}[0]{\text{\normalfont{LTL}}}

\newcommand{\hyperltl}[0]{\text{\normalfont{HyperLTL}}}
\newcommand{\hyperctls}[0]{\text{\normalfont{HyperCTL$^*$}}}

\newcommand{\mon}{\oplus}
\newcommand{\antimon}{\ominus}
\newcommand{\checkfun}{\lstinline[style=default,language=custom-lang]|check|}
\newcommand{\moso}{\texttt{MoSo}}

\usetikzlibrary{arrows,fit,shapes.geometric,backgrounds}
\definecolor{dkcyan}{rgb}{0.1, 0.3, 0.3}
\colorlet{comment-color}{black!50}

\lstdefinelanguage{custom-lang}{
	keywords={let, in, match, with, when, if, then, else, elif, for, to, do, rec, return, new, not, and, while,def},
	keywordstyle=[1]\bfseries,
	morekeywords=[2]{computeFix,check,Monitor,computeMonotonicity,getNextTrace,order},
	keywordstyle=[2]\color{dkcyan},
        morekeywords=[3]{SAT,UNSAT,true,false},
	keywordstyle=[3]\color{black!70},
	comment=[l][\color{comment-color}]{//},
	literate=%
	{=}{{{=}}}1
	{<-}{{{$\leftarrow$}}}1
	{|}{{{$\mid$}}}1
	{:}{{{:}}}1
	{:=}{{{:=}}}1
	{@}{ }1
}

\lstdefinestyle{default}{
	escapeinside={(*}{*)},
	basicstyle=\ttfamily\fontsize{7.5}{9}\selectfont,
	columns=fullflexible,
	commentstyle=\sffamily\color{black!50!white},
	framexleftmargin=1em,
	framexrightmargin=1ex,
	keepspaces=true,
	keywordstyle=\color{dkblue},
	mathescape,
	numbers=left,
	numberblanklines=false,
	numbersep=0.5em,
	numberstyle=\relscale{0.75}\color{gray}\ttfamily,
	showstringspaces=true,
	stepnumber=1,
	xleftmargin=1.2em,
}

\lstdefinestyle{defaultH}{
	backgroundcolor=\color{green},
	escapeinside={(*}{*)},
	basicstyle=\ttfamily\fontsize{7.5}{9}\selectfont,
	columns=fullflexible,
	commentstyle=\sffamily\color{black!50!white},
	framexleftmargin=1em,
	framexrightmargin=1ex,
	keepspaces=true,
	keywordstyle=\color{dkblue},
	mathescape,
	numbers=left,
	numberblanklines=false,
	numbersep=0.5em,
	numberstyle=\relscale{0.75}\color{gray}\ttfamily,
	showstringspaces=true,
	stepnumber=1,
	xleftmargin=1.2em,
}

\usepackage{tcolorbox}

\lstnewenvironment{mycode}[1][]
{\small
	\lstset{
		style=default, 
		language=custom-lang,
		#1
	}
}
{}

\newcommand\xqed[1]{%
	\leavevmode\unskip\penalty9999 \hbox{}\nobreak\hfill
	\quad\hbox{#1}}
\newcommand\demo{\xqed{$\triangle$}}

\makeatletter
\gdef\@copyrightpermission{
	\begin{minipage}{0.3\columnwidth}
		\href{https://creativecommons.org/licenses/by/4.0/}{\includegraphics[width=0.90\textwidth]{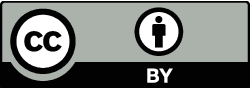}}
	\end{minipage}\hfill
	\begin{minipage}{0.7\columnwidth}
		\href{https://creativecommons.org/licenses/by/4.0/}{This work is licensed under a Creative Commons Attribution International 4.0 License.}
	\end{minipage}
	\vspace{5pt}
}

\begin{document}

\pagestyle{fancy}
\fancyhead{}

\maketitle

\section{Introduction}\label{sec:intro}
Monitoring is a practical and scalable method for ensuring that the behavior of a complex system satisfies its formal specification. Unlike traditional verification techniques, such as model checking and theorem proving, monitoring
does not analyze a system model, but instead works directly with the traces of the running system.  
Monitoring has been well studied for various types of trace properties, such as invariants or other properties expressible in standard temporal logics like linear-time temporal logic (LTL)~\cite{DBLP:journals/entcs/FinkbeinerS01,DBLP:journals/sttt/HavelundR04,10.1145/2000799.2000800}. 
By contrast, the study of monitoring algorithms for hyperproperties is still in an early stage. 

Hyperproperties express relationships between multiple traces. Many properties related to knowledge, information flow, and privacy are hyperproperties and can, therefore, not be analyzed with methods for trace properties. The conceptual challenge is that while trace properties are sets of traces, hyperproperties are \emph{sets of} sets of traces. So while the monitor of a trace property simply decides whether the observed trace is an element of the given set, the monitor of a hyperproperty must consider the entire set of traces produced by the system under observation (including any traces that will only be produced in the future) and decide whether this full set of traces is an element of the specified set of sets of traces.

Previous work on monitoring hyperproperties~\cite{conf/csfw/AgrawalB16,conf/tacas/BrettSB17,conf/csfw/BonakdarpourF18,DBLP:journals/fmsd/FinkbeinerHST19} has focused on the special case of first-order hyperproperties, in particular on properties that can be expressed in the temporal logic HyperLTL~\cite{ClarksonFKMRS14}. HyperLTL extends LTL with first-order quantification over traces. First-order hyperproperties include many information-flow policies, such as noninterference and some simple notions of knowledge. Consider, for example, the property that an agent~$i$ in a multi-agent system (MAS) knows that an LTL formula $\varphi$ is true whenever formula $\varphi$ is actually true. \emph{Knowing} that $\varphi$ holds on a trace $\pi$ means that $\varphi$ must hold on all traces $\pi'$ that are indistinguishable from $\pi$ for agent~$i$. The property is, therefore, expressed by the HyperLTL formula
\[\forall \pi. \forall \pi'.\ \varphi[\pi] \rightarrow (\pi \sim_i \pi' \rightarrow \varphi[\pi']),\]
where $\pi \sim_i \pi'$ denotes that agent $i$ cannot distinguish between $\pi$ and $\pi'$, and $\varphi[\pi]$ encodes that $\varphi$ holds on trace $\pi$.
A HyperLTL monitor for this formula detects a violation as soon as a trace $\pi'$ is observed that is indistinguishable from some previously seen trace $\pi$ and yet $\pi'$ does not satisfy $\varphi$.

Many hyperproperties of interest, however, cannot be
expressed with first-order quantifiers alone. A prominent example is common knowledge \cite{DBLP:books/mit/FHMV1995,HalpernM90,Meyden98}.
\emph{Common knowledge} (CK) requires that every agent~$i$ not only knows that $\varphi$ is true, but, additionally, knows that every other agent knows that agent~$i$ knows that $\varphi$ is true. The fact that every agent knows this, must, in turn, be known by every other agent, and so on. 
Formally, common knowledge refers to a set of traces that is closed under the indistinguishability relations of all agents, and requires that $\varphi$ holds on all traces in this set.
\sohyperltl{} \cite{BeutnerFFM23} extends \hyperltl{} with second-order quantification over sets of traces and can thus express properties like CK. 
Since model-checking \sohyperltl{} formulas is undecidable, the verification method for \sohyperltl{} is based on approximation techniques. The question arises whether or not monitoring second-order hyperproperties suffers from a similar limitation.  
In this paper, we show that while the precise answer depends on the underlying execution model, for a large class of second-order hyperproperties the monitoring problem can, in fact, be solved effectively.

As a specification logic for monitoring second-order hyperproperties in MASs, we
introduce the temporal logic \sohyperltlstar{}. The main difference to \sohyperltl{} is that 
monitoring deals with \emph{finite} rather than infinite traces; \sohyperltlstar{} therefore has a finite-trace semantics. 
\sohyperltlstar{} furthermore includes past-time temporal operators and allows for arbitrary nestings of temporal operators and first/second-order quantifiers. Common knowledge in a MAS with agents  $1, \ldots, n$ can be expressed as the  \sohyperltlstar{} formula:
\begin{align*}
    &\forall \pi\ldot \varphi[\pi] \to \exists X \ldot \pi \in X \, \land \\
    &\quad\Big( \forall \pi_1 \in X.\, \forall \pi_2.\ \big(\bigvee_{i=1}^n \pi_1 \sim_i \pi_2 \big) \rightarrow \pi_2 \in X \Big)\, \land\, \forall \pi' \in X\ldot \varphi[{\pi'}].
\end{align*}
The second-order quantifier $\exists X$ specifies the existence of a set $X$ of traces so that $\pi$ is in $X$, and all traces that are indistinguishable from some trace in $X$ for some agent are also in $X$.
Given this set, a monitor detects a violation by finding a trace in the set that does not satisfy $\varphi$.

Two fundamental execution models in the monitoring of hyperproperties are the parallel and sequential models (cf.~\cite{DBLP:journals/fmsd/FinkbeinerHST19}): in the \emph{parallel model}, a fixed number of traces is monitored in parallel; in the \emph{sequential model}, an unbounded number of traces is observed sequentially, one trace after the other.

For the parallel model, we show that the monitoring of the second-order hyperproperties expressed in \sohyperltlstar{} can be reduced to monitoring first-order hyperproperties. 
For the sequential model, we show that the monitoring problem is undecidable in general, but becomes feasible for the practically relevant class of monotone second-order hyperproperties.
A second-order hyperproperty is $\mon$-monotone if its satisfaction on some set of traces implies that it is also satisfied on any superset of this set. Conversely, a second-order hyperproperty is 
$\antimon$-monotone if a violation on some set implies its violation on any superset. Monotonicity thus allows the monitor to provide \emph{definitive} results that hold irrespectively of the traces that may still arrive in the future. 

We introduce an inference system for monotonicity and present a monitoring algorithm that iteratively checks the given set of traces until it can produce a decisive answer. 
The algorithm has been implemented in a tool called \moso{}.
We report on encouraging evaluation results for \moso{} on several benchmarks, including examples from common knowledge and planning.

\paragraph{Related Work}
\citet{DBLP:conf/ijcai/GiacomoV13} first introduced a logic for interpreting \ltl{} on finite traces, called \ltlf{}, specifically designed for AI systems.
Since then, \ltlf{} and its variants have been used, e.g., for model checking \cite{DBLP:journals/corr/abs-2305-08319,RajasekaranV22}, satisfiability analysis~\cite{DBLP:conf/aaai/FiondaG16}, synthesis~\cite{DBLP:conf/ijcai/GiacomoV15,GutierrezPW21}, and  planning~\cite{DBLP:conf/aips/CamachoBMM18}.
Knowledge in combination with \ltl{} dates back to \citet{DBLP:books/mit/FHMV1995} and has has found many applications in MASs~\cite{DBLP:conf/ijcai/KongL17, DBLP:conf/atal/KongL18, DBLP:conf/aaai/FelliMPW23}.
Logics for hyperproperties have mostly been obtained by extending temporal logics with explicit quantification over traces or paths, such as HyperLTL~\cite{ClarksonFKMRS14}, HyperQPTL \cite{DBLP:phd/dnb/Rabe16}, HyperPDL~\cite{GutsfeldMO20}, HyperATL$^*$~\cite{BeutnerF21,BeutnerF24}, and HyperLDL$_f$~\cite{GiacomoFMP21}.
\sohyperltl{} adds second-order quantification over \emph{sets} of traces \cite{BeutnerFFM23}.
Tools for model-checking knowledge in multi-agent systems include \texttt{MCK}~\cite{DBLP:conf/cav/GammieM04} and \texttt{MCMAS}~\cite{DBLP:journals/sttt/LomuscioQR17}, the latter was also extended to finite trace semantics \cite{DBLP:conf/atal/KongL18}.
While these approaches and tools implement solutions explicitly for knowledge, \sohyperltlstar{} generalizes to more general second-order hyperproperties.

\section{\titlesohyperltlstar}\label{sec:sohyperltlstar}

We consider a finite set of atomic propositions $AP$ and define $\Sigma := 2^{AP}$. 
We define \sohyperltlstar\ as a finite-trace extension of \sohyperltl\ \cite{BeutnerFFM23}.
Let $\tracevars{} = \{\pi, \pi', \pi_1, \ldots \}$ be a set of (first-order) trace variables and $\sovars{} = \{X, Y, \ldots\}$ a set of (second-order) set variables.
We assume a special second-order variable $\systemvar \in \sovars{}$ that we use to refer to the set of all system traces.
\sohyperltlstar{} formulas are defined by the following grammar:
\begin{align*}
\varphi :=\ a_\pi &\mid \neg \varphi \mid \varphi \land \varphi \mid \X \varphi \mid \LTLprevious \varphi \mid \varphi \U \varphi \mid \varphi \LTLsince \varphi \\
&\mid \quant \pi \in X \ldot \varphi \mid \quant X\ldot \varphi
\end{align*}
where $a \in \ap{}$, $\pi \in \tracevars{}$, $X \in \sovars{}$, and $\quant \in \{\forall, \exists\}$ is a quantifier.
In \sohyperltlstar{}, we have the future temporal operators (strong) \emph{next} $\X$ and \emph{until} $\LTLuntil$, as well as their past counterparts \emph{previously} $\LTLprevious$ and \emph{since} $\LTLsince$.
We use the derived Boolean constants $\mathit{true}, \mathit{false}$, and connectives $\vee, \rightarrow, \leftrightarrow$, and the temporal operators \emph{eventually} $\LTLeventually \varphi = \mathit{true} \LTLuntil \varphi$, \emph{once} $\LTLpastfinally \varphi = \mathit{true} \LTLsince \varphi$, \emph{globally} $\LTLglobally \varphi = \neg \LTLeventually \neg \varphi$, and \emph{historically} $\LTLpastglobally \varphi = \neg \LTLpastfinally \neg \varphi$.

\paragraph{Semantics}
The semantics of \sohyperltlstar{} is defined with respect to a trace length $m \in \nat$ and a set of traces $\traceSet \subseteq \Sigma^m$.\footnote{We restrict to traces of the same length to obtain simpler semantics. Without this restriction, we would have to deal with combinations of traces with different lengths, making traversal difficult. In practice, we can either pad traces to the length of the longest trace or crop them to the length of the shortest trace (cf.~\cite{DBLP:conf/tacas/FinkbeinerHST18,DBLP:journals/fmsd/FinkbeinerHST19}).}
As for HyperLTL, we use a trace assignment $\Pi : \tracevars \rightharpoonup \Sigma^m$ mapping trace variables in $\tracevars$ to finite traces of length $m$. 
Additionally, we maintain a second-order assignment $\Delta : \sovars{} \rightharpoonup 2^{\Sigma^m}$ mapping variables to sets of finite traces of length $m$.
The semantics is then as follows:
\begin{align*}
	\Pi, \Delta, i &\models_\traceSet  a_\pi &\text{iff} \quad  &a \in \Pi(\pi)(i)\\
	\Pi, \Delta, i &\models_\traceSet  \neg \varphi &\text{iff} \quad & \Pi, \Delta, i \not\models_\traceSet  \varphi \\
	\Pi, \Delta, i &\models_\traceSet  \varphi_1 \land \varphi_2 &\text{iff} \quad  &\Pi, \Delta, i \models_\traceSet \varphi_1 \text{ and }  \Pi, \Delta, i \models_\traceSet  \varphi_2\\
	\Pi, \Delta, i &\models_\traceSet  \X  \varphi &\text{iff} \quad &i < m - 1 \text{ and }
\Pi, \Delta, i+1 \models_\traceSet \varphi \\
    \Pi, \Delta, i &\models_\traceSet  \LTLprevious  \varphi &\text{iff} \quad & i > 0 \text{ and } \Pi, \Delta, i-1 \models_\traceSet \varphi \\
	\Pi, \Delta, i &\models_\traceSet  \varphi_1 \U \varphi_2 &\text{iff} \quad & \exists i \leq j < m\ldot \Pi, \Delta, j\models_\traceSet  \varphi_2 \text{ and }\\ 
    &&&\quad\forall i \leq k < j \ldot  \Pi, \Delta, k \models_\traceSet  \varphi_1\, \\
    \Pi, \Delta, i &\models_\traceSet  \varphi_1 \LTLsince \varphi_2 &\text{iff} \quad & \exists j \leq i\ldot  \Pi, \Delta, j\models_\traceSet  \varphi_2 \text{ and } \\ 
    &&&\quad\forall i \geq k > j \ldot  \Pi, \Delta, k \models_\traceSet  \varphi_1\, \\
    \Pi, \Delta, i &\models_\traceSet  \quant \pi \in X \ldot \varphi &\text{iff} \quad &\quant t \in \Delta(X) \ldot \Pi[\pi \mapsto t], \Delta, i \models_\traceSet  \varphi\\
	\Pi, \Delta, i &\models_\traceSet  \quant X \ldot \varphi &\text{iff} \quad &\quant A \subseteq \traceSet \ldot \Pi, \Delta[X \mapsto A], i \models_\traceSet  \varphi
\end{align*}
The temporal and boolean operators are evaluated as expected.
Atomic formula $a_\pi$ holds in step $i$ if $a$ holds in the $i$the step on the trace bound to $\pi$.
When quantifying over a trace $\pi \in X$, we quantify over all traces in its current second-order model (as given by $\Delta$), and when quantifying over a set of traces, we consider all possible subsets of $\traceSet$. 
A set of traces $\traceSet \subseteq \Sigma^m$ satisfies $\varphi$, written $\traceSet \models \varphi$, if $\emptyset, [\systemvar \mapsto \traceSet], 0 \models_\traceSet \varphi$.
Note that we bind the second-order variable $\systemvar$ to $\traceSet$; by writing $\forall \pi \in \systemvar$ and $\exists \pi \in \systemvar$ we can thus quantify over the traces in $\traceSet$ (as in HyperLTL).

\paragraph{HyperLTL and \sohyperltlstar{}}
There are four key differences between the logics \sohyperltlstar{} and \hyperltl{}:
Firstly, \sohyperltlstar{} adds past operators, which allow for exponentially more succinct definitions of properties~\cite{DBLP:journals/eatcs/Markey03} and offers a convenient syntax to specify knowledge properties (cf.~\Cref{ex:f-knowledge}).
Secondly, the semantics of \sohyperltlstar{} is defined w.r.t.~finite traces, which occur much more frequently in MAS-related domains \cite{DBLP:conf/ijcai/GiacomoV15,DBLP:journals/corr/abs-2305-08319,DBLP:conf/ijcai/GiacomoV13,DBLP:books/mit/FHMV1995}.
Thirdly, \sohyperltlstar{} allows quantification under temporal operators.
In \hyperltl{} \cite{ClarksonFKMRS14} or \sohyperltl{} \cite{BeutnerFFM23}, a formula consists of a quantifier prefix followed by a quantifier-free LTL body. 
In contrast, in \sohyperltlstar{}, we can interleave temporal operators and quantification. 
For example, $\forall \pi \in \systemvar\ldot \LTLeventually \forall \pi' \in \systemvar\ldot \LTLglobally(a_{\pi} \leftrightarrow a_{\pi'})$ states that on any trace, we can find a timepoint from which point on all other traces agree in $a$.
Note that despite interleaving temporal operators and quantification, \sohyperltlstar{} has a \emph{linear-time semantics} (different from, e.g., \hyperctls{} \cite{ClarksonFKMRS14}), which is crucial in our monitoring setting. 
And lastly, and most importantly, \sohyperltlstar{} features second-order quantification over \emph{sets of traces}. 
Note that while second-order quantification in \sohyperltl{} \cite{BeutnerFFM23} ranges over arbitrary sets of traces, \sohyperltlstar{} quantifies over subsets of $\traceSet$.

\begin{example}[Eventual Knowledge]\label{ex:f-knowledge}
	As a running example, we use the sender-receiver MAS modeled by transition system in \Cref{fig:sender:receiver}, which is a variation of an example presented in~\cite{Meyden98}. 
In this system, agent 1 attempts to send a message until it is either received by agent 2, or delayed for one step and received afterward.
The two agents each only have a partial view of the system: agent 1 observes the sending state (i.e., only observes AP $s$) and thus cannot distinguish if the message is immediately received or delayed for one step.
Likewise, agent 2 only observes if the message was received (modeled by AP $r$).
We want to express that whenever a message is received twice (i.e., two consecutive $r$s occur), agent 1 eventually knows that it was received.
As in epistemic logics, an agent knows a property on a trace if it holds on all indistinguishable traces, which we can express as follows:
\begin{align*}
\forall \pi \in \systemvar. \big(\LTLeventually (&r_\pi \wedge \LTLnext r_\pi) \big) \rightarrow \LTLeventually(\forall \pi' \in \systemvar\ldot \pi \sim_1 \pi' \to \LTLeventually r_{\pi'}),
\end{align*}
where we define %
$\pi \sim_1 \pi' := \LTLpastglobally (s_\pi \leftrightarrow s_{\pi'})$
to state that agent 1 cannot distinguish $\pi$ and $\pi'$ up to the current time step.
Here, we quantify over all traces $\pi$ that visit the receiving state twice.
For any such $\pi$, we require that \emph{eventually}, all traces $\pi'$ that are -- up to the current time point -- indistinguishable from $\pi$, satisfy $\LTLeventually r_{\pi'}$.
This property holds on the system in \Cref{fig:sender:receiver}.
For example, consider the trace $s^{n}rr$, where sending takes $n$ steps and immediately afterward the message is received. 
In the last time step, this trace is indistinguishable (for agent 1) from $s^{n}rr$ and $s^{n}dr$, and both of these traces satisfy $\LTLeventually r$.\demo
\end{example}
\begin{figure}
    \centering
    \resizebox{.7\linewidth}{!}{
    \tikzstyle{state}=[draw, circle, fill=none, inner sep=1mm, thick]

\begin{tikzpicture}[->,>=stealth',shorten >= 1pt,auto]

\node[state, scale=1.1, minimum size=18pt] (s1) [] {%
   \small$d$
 };

\node[state, scale=1.1, minimum size = 18pt] (s0) [left = 0.2 and 1 of s1]{%
   \small$s$
};

\node (left) [left = 0.5 of s0, draw=none]{};

\node[state, scale=1.1, minimum size = 18pt] (s2) [right = 0.2 and 1 of s1] {%
    \small$r$
 };

\begin{scope}[on background layer]
	\node[fit=(s0) (s1), ellipse, draw=orange, very thick, dashed, fill=orange!20, yshift=0.18cm, fill opacity = 0.5, minimum height=44pt] (s0s1) {};
	\node[fit=(s1) (s2), ellipse, rotate=0, very thick, dotted, yshift=0.18cm, fill=cyan!20, draw=cyan, fill opacity=0.5, minimum height=44pt] (s1s2) {};
\end{scope}

\path (left) edge[thick] (s0)
      (s0) edge[thick, bend left=0, align=center] node[above, sloped] {%
      } (s1)
      (s0) edge[loop above, thick,  align=center] node[above,sloped] {%
      } (s0)
      (s0) edge[thick, bend left, align=center] node[sloped,below] {%
      } (s2)
      (s1) edge[thick, bend right=0, align=center] node[sloped,below] {%
      } (s2)
      (s2) edge[loop above, thick, align=center] node[sloped,below] {%
      } (s2)
      ;
\end{tikzpicture}
    }
    \caption{A transition system modelling a sender-receiver MAS. Agent 1 only observes AP $s$, and thus cannot distinguish between the states in the blue (dotted) area. %
    Agent 2 only observes %
    AP $r$ and thus cannot distinguish between the states in the orange (dashed) area.%
    }
    \label{fig:sender:receiver}
\end{figure}
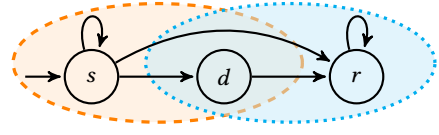

\paragraph{Fixpoint Formulas}

As already observed in \cite{BeutnerFFM23}, full second-order quantification is hard to handle algorithmically, particularly in a monitoring setting:
If the monitor has seen $n$ finite traces, any of the $2^n$ subsets might qualify as a witness and need to be checked. 
Fortunately, full second-order quantification is not needed for most properties of interest. 
Instead, most second-order sets are defined in terms of a monotone fixpoint.
Similar to \cite{BeutnerFFM23}, we extend the syntax of \sohyperltlstar{} with a dedicated fixpoint construct, which allows for a convenient definition of such fixpoint-based sets of traces. 
Moreover, our monitoring algorithm can exploit the monotonicity properties of fixpoints. 
The extended syntax of \sohyperltlstar{} is then the following:
\begin{align*}
\varphi := a_\pi &\mid \neg \varphi \mid \varphi \land \varphi \mid \X \varphi \mid \LTLprevious \varphi \mid \varphi \U \varphi \mid \varphi \LTLsince \varphi \\
& \mid \quant \pi \in X \ldot \varphi \mid \quant X\ldot \varphi \mid \mathit{fix}(X, \xi_1, \ldots, \xi_k)\ldot \varphi.
\end{align*}
The fixpoint operator $\mathit{fix}(X, \xi_1, \ldots, \xi_k)\ldot \varphi$ constructs a unique set of traces $X \in \sovars{}$ that can be used in $\varphi$.
This set $X$ is uniquely defined by the fixpoint constraints $\xi_1, \ldots, \xi_k$ of the form 
\begin{align}\label{eq:fix}
    \forall \pi_1 \in X_1\ldots \forall \pi_n \in X_n\ldot \varphi_\mathit{step} \to \pi \in X,
\end{align}
where $X_1, \ldots, X_n \in \sovars{}$, $\pi, \pi_1, \ldots, \pi_n \in \tracevars{}$, and $\varphi_\mathit{step}$ is a quantifier-free formula.
To be well-formed, (1) all trace variables used in $\varphi_\mathit{step}$ must either be quantified outside or be one of $\pi_1, \ldots, \pi_n$; (2) all second-order variables $X_1, \ldots, X_n$ must be quantified outside or be equal to $X$, i.e., within the definition of $X$ we can quantify over traces in $X$ (as is usual for fixpoint definitions); and (3) $\pi$ must be quantified outside or be one of $\pi_1, \ldots, \pi_n$.
Intuitively, \Cref{eq:fix} states a requirement on traces that should be included in $X$. 
If we find traces $t_1 \in X_1, \ldots, t_n \in X_n$ that, together with the sets and traces quantified outside of $\mathit{fix}(X, \xi_1, \ldots, \xi_k)\ldot \varphi$, satisfy $\varphi_\mathit{step}$, then trace $\pi$ should be added to $X$.
In our semantics, we define $\Pi, \Delta, i \models_\traceSet \mathit{fix}(X, \xi_1, \ldots, \xi_k)\ldot \varphi$ iff
\begin{align*}
    \Pi, \Delta\big[X \mapsto \mathit{sol}(\Pi, \Delta, i, X, \xi_1, \ldots, \xi_k)\big], i \models_\traceSet  \varphi,
\end{align*}
where $\mathit{sol}(\Pi, \Delta, i, X, \xi_1, \ldots, \xi_k)$ denotes the unique solution to the fixpoint definition of $X$. 
Formally, $\mathit{sol}(\Pi, \Delta, i, X, \xi_1, \ldots, \xi_k)$ is the \emph{smallest} set of traces such that \emph{for all} $\xi_i$ we have
\begin{align*}
    \Pi, \Delta[X \mapsto \mathit{sol}(\Pi, \Delta, i, X, \xi_1, \ldots, \xi_k)], i \models_\traceSet \xi_i.
\end{align*}
Note that the fixpoint solution $\mathit{sol}(\Pi, \Delta, i, X, \xi_1, \ldots, \xi_k)$ is uniquely defined.\footnote{Given $\Pi, \Delta, i$, let $f : 2^\traceSet \to 2^\traceSet$ be the function that -- given a current model for X -- returns X with all traces that should be added according to one of the $\xi_i$.
It is easy to see that $f$ is monotone in the subset order on $2^\traceSet$.
Each fixpoint of $f$ now corresponds to a set $X$ that satisfies $\xi_1, \ldots, \xi_k$.
By Knaster-Tarski \cite{tarski1955lattice}, there exists a \emph{unique} smallest fixpoint of $f$, which is exactly $\mathit{sol}(\Pi, \Delta, i, X, \xi_1, \ldots, \xi_k)$. }
Using this fixpoint construct, we can easily express CK:

\begin{example}[Common Knowledge and Fixpoints]\label{ex:common:knowledge:fixpoint}
We illustrate our fixpoint construct by continuing on \Cref{ex:f-knowledge}.
Assume we do not want to state that on all traces where $\LTLeventually (r \wedge \LTLnext r)$ holds, agent 1 knows that $\LTLeventually r$ holds (cf.~\Cref{ex:f-knowledge}), but rather that it is \emph{common knowledge} in the group of agents 1 and 2 that $\LTLeventually r$ holds.
As argued in \Cref{sec:intro}, this requires us to iteratively compute the set of all traces that cannot be distinguished by some agent.
While the formulation of CK in \Cref{sec:intro} expressed this using full second-order quantification, we can observe that the set we are interested in is actually defined by a fixpoint:
\begin{align*}
    &\forall \pi \in \systemvar. \big(\LTLeventually (r_\pi \wedge \LTLnext r_\pi) \big) \rightarrow \LTLeventually \Big(\mathit{fix}(X, \xi_1, \xi_2) \ldot \forall \pi' \in X\ldot \LTLeventually r_{\pi'} \Big),
\end{align*}
where $\xi_1 := \mathit{true} \to \pi \in X$, and $\xi_2$ is defined as 
\begin{align*}
    \forall \pi_1 \in X. \forall \pi_2 \in \systemvar{}\ldot \big( \pi_1 \sim_1 \pi_2 \lor \pi_1 \sim_2 \pi_2 \big)\to \pi_2 \in X.
\end{align*}
Here, we define $\pi \sim_2 \pi' := \LTLpastglobally (r_\pi \leftrightarrow r_{\pi'})$, similar to $\sim_1$ in \Cref{ex:f-knowledge}.
For each $\pi$, we specify a fixpoint-defined set $X$. 
This set $X$ includes $\pi$ (constraint $\xi_1$), and whenever we find some $\pi_1 \in X$ that some agent cannot distinguish from some $\pi_2 \in \systemvar$, we add $\pi_2$ to $X$ (constraint $\xi_2$).
Then, we state that all traces $\pi'$ in $X$ should satisfy $\LTLeventually r_{\pi'}$.
The above formula does not hold on the MAS modeled in \Cref{fig:sender:receiver}. 
Consider the trace $\pi = sr^n$.
It is easy to see that
\[sr^n \sim_1 sdr^{n-1} \sim_2 s^2r^{n-1} \sim_1 s^2dr^{n-2} \sim_2 \ldots \sim_2 s^{n}r \sim_1 s^{n}d.\]
Consequently, the trace $s^{n}d$ will be included in $X$, disproving that $\LTLeventually r$ is common knowledge on $\pi$.\demo
\end{example}

\section{Monitorability}

We are interested in \emph{monitoring} a \sohyperltlstar{} formula, i.e., we do not have access to the system as a whole but rather observe executions of the system and conclude whether the set of traces we have seen so far suffices to conclude the satisfaction or violation of the property. 
This is, unsurprisingly, not possible for all properties.
Some properties are not \emph{monitorable}, i.e., never allow a positive or negative answer from the monitor.

\subsection{Monitorability of Trace Properties}

Let us recall the concept of monitorability in the simpler setting of a trace property, where we observe one step of the trace at a time.

\begin{definition}\label{def:mon}
    A trace property  $L \subseteq \Sigma^*$ is monitorable if 
\begin{align*}
    \forall u \in \Sigma^*\ldot \exists v \in \Sigma^*\ldot (\forall w \in \Sigma^*\ldot uvw \in L) \lor (\forall w \in \Sigma^*\ldot uvw \not\in L).
\end{align*}
\end{definition}

Intuitively, the definition states that -- whatever finite trace $u$ we start from -- some extension $uv$ of $u$ allows the monitor to report the satisfaction or violation of $L$. 
This can either be the case because all extensions of $uv$ satisfy $L$ (the first disjunct) or all extensions violate $L$ (the second disjunct).\footnote{The sets $\mathit{good}(L) := \{u \in \Sigma^* \mid \forall v \in \Sigma^*\ldot uv \in L\}$ and $\mathit{bad}(L) := \{u \in \Sigma^* \mid \forall v \in \Sigma^*\ldot uv \not\in L\}$ described by the disjunctions are commonly referred to as the set of good and bad prefixed, respectively. }
In this paper, we are not monitoring a trace property but a hyperproperty, i.e., we may need to observe multiple traces.
Depending on the setting in which the monitor is deployed, there can be different modes in which those traces are presented to the monitor. 
We focus on the \emph{parallel model} and the \emph{sequential model} \cite{DBLP:journals/fmsd/FinkbeinerHST19}.
In the former, the number of traces is fixed \emph{a priori}, and the monitor observes consecutive time steps one by one.
In the latter model, the monitor observes the (finite) traces sequentially, increasing the cardinality of the trace set in every time step.
We discuss both models and their impact on the monitorability of \sohyperltlstar{} properties in the following.

\subsection{The Parallel Model}

\newcommand{\unfold}[1]{\llbracket #1 \rrbracket}
In the parallel model, the number of traces (executions) is fixed to some number $b \in \nat$, and each time step reveals an additional position on each of the traces. 
Similar to the monitorability of trace properties, the monitor thus gets presented more and more positions on the same $b$ traces.
We refer to~\cite{DBLP:journals/fmsd/FinkbeinerHST19} for more details.

As the parallel model \emph{fixes} the number of sessions (traces), it turns out that second-order quantification does not yield any additional expressiveness.
Intuitively, quantifying over a set of traces is equivalent to quantifying over some finite subset of at most $b$ traces, which is expressible in HyperLTL with first-order trace quantification. 
Formally, given bound $b \in \nat$ and a \sohyperltlstar{} formula $\varphi$, we can unfold $\varphi$ into a formula that uses no second-order quantification and is equivalent on any set of traces with at most $b$ traces.
We maintain a partial function $M : \sovars{} \rightharpoonup 2^\tracevars{}$ mapping second-order variables to a set of trace variables. 
We then recursively define the \sohyperltlstar{} formula $\unfold{\varphi}_{b, M}$ as follows:
\begin{align*}
    \unfold{a_\pi}_{b, M} &:= a_\pi\\
    \unfold{\circ \varphi}_{b, M} &:= \circ \unfold{\varphi}_{b, M} \quad \text{for } \circ \in \{\neg, \X, \LTLprevious\}\\
    \unfold{\varphi_1 \circ \varphi_2}_{b, M} &:= \unfold{\varphi_1}_{b, M} \circ \unfold{\varphi_2}_{b, M} \quad \text{for } \circ \in \{\land, \LTLuntil, \LTLsince\}\\
    \unfold{\exists X\ldot \varphi}_{b, M} &:= \exists \pi_1, \ldots, \pi_b\ldot \unfold{\varphi}_{b, M[X \mapsto \{\pi_1, \ldots, \pi_b\}]}\\
    \unfold{\forall X\ldot \varphi}_{b, M} &:= \forall \pi_1, \ldots, \pi_b\ldot \unfold{\varphi}_{b, M[X \mapsto \{\pi_1, \ldots, \pi_b\}]}\\
    \unfold{\exists \pi \in X\ldot \varphi}_{b, M} &:= \bigvee_{\pi' \in M(X)} \unfold{\varphi}_{b, M}[\pi'/\pi]\\
    \unfold{\forall \pi \in X\ldot \varphi}_{b, M} &:= \bigwedge_{\pi' \in M(X)} \unfold{\varphi}_{b, M}[\pi'/\pi]
\end{align*}
Here, $\varphi[\pi'/\pi]$ denotes the formula in which all free occurrences of $\pi$ are replaced with $\pi'$.
For quantifier-free formulas, we simply maintain the structure. 
Since we assume that there are at most $b$ traces, we can replace second-order quantification over a set $X$ with first-order quantification over $b$ \emph{fresh} (first-order) trace variables $\pi_1, \ldots, \pi_b$.
We record the first-order trace variables that we used to replace $X$ within the auxiliary mapping $M$.
For first-order quantification $\forall \pi \in X$ (resp.~$\exists \pi \in X$), we then conjunctively (resp.~disjunctively) consider all traces $\pi' \in M(X)$ in place of $\pi$.
A simple induction shows:

\begin{proposition}\label{prop:so-free}
    Let $\varphi$ be any \sohyperltlstar{} formula and $b \in \nat$. For any set of trace $\mathbb{T}$ with $|\mathbb{T}| \leq b$ we have $\mathbb{T} \models \varphi$ iff $\mathbb{T} \models \unfold{\varphi}_{b, \emptyset}$.
\end{proposition}
Assuming a fixed bound $b \in \nat$ (as in the parallel model), \Cref{prop:so-free} states that monitoring a \sohyperltlstar{} formula $\varphi$ is equivalent to monitoring $\unfold{\varphi}_{b, \emptyset}$.
As $\unfold{\varphi}_{b, \emptyset}$ uses \emph{no} second-order quantification, we can apply all the techniques and tools developed for first-order logics like \hyperltl{} \cite{DBLP:journals/fmsd/FinkbeinerHST19}.

\subsection{The Sequential Model}

We now consider the sequential model, the main object of study in this paper. 
In this model, traces arrive one at a time and the number of traces is \emph{unbounded}.
Our model is similar to the sequential model used in \cite{DBLP:journals/fmsd/FinkbeinerHST19} for \hyperltl{}, with the exception that traces in our model are of finite (instead of infinite) length. 
This focus on finite traces makes our approach (and tool) applicable to many MAS-related settings, where executions are finite, and tasks are subsequent system executions.

\begin{definition}\label{def:hypermonitorabilitysequential}
A hyperproperty  $H \subseteq 2^{\Sigma^*}$ is monitorable in the sequential model if
\begin{align*}
	&\forall U \subseteq \Sigma^*\ldot \exists V \subseteq \Sigma^*.\\
    &\quad\big(\forall W \subseteq \Sigma^*\ldot (U \cup \! V  \cup  W) \in H\big) \lor \big(\forall W \subseteq \Sigma^*\ldot (U  \cup V \cup W) \not\in H\big).
\end{align*}

\end{definition}
That is, for any $U$ there exists some $V$ such that $U \cup V$ allows some definitive answer by the monitor. 
This can either be the case because all extensions of $U \cup V$ satisfy $H$ (the first disjunct) or all extensions violate $H$ (the second disjunct).

\begin{example}
Consider the CK formula in \Cref{ex:common:knowledge:fixpoint}.
The formula is monitorable in the sequential model:
We can always add traces to ensure that CK does not hold, i.e., add some indistinguishable trace that does not satisfy $\LTLeventually r$.
No matter what additional traces are observed, CK remains violated.\demo

\end{example}

\begin{theorem}\label{th:undecidability}
Deciding if a \sohyperltlstar{} formula is monitorable in the sequential model is undecidable.
\end{theorem}
\begin{proof}
\citet{DBLP:journals/fmsd/FinkbeinerHST19} showed that monitorability is undecidable for \hyperltl{} with finite traces in the sequential model. 
Since \hyperltl{} with finite-trace semantics is a strict subset of \sohyperltlstar, monitorability of \sohyperltlstar{} is undecidable. 
\end{proof}

\section{Monotonicity}\label{sec:monotonicity}
In the following, we focus on the sequential model.
\Cref{th:undecidability} rules out the possibility of an algorithm that monitors all possible \sohyperltlstar{} formulas. 
Instead, we focus on fragments of formulas for which we can provide definitive answers. 
Given a finite set of traces, we want to provide a monitoring result irrespective of what traces will arrive in the future. 
To ensure this, the truth value of the formula may not change when additional traces arrive; that is, it should be \emph{monotone}.
We distinguish between $\mon$-monotonicty and $\antimon$-monotonicty.
In $\mon$-monotonicty, any model that \emph{satsfies} the formula will continue to satisfy it no matter what traces are added. 
In $\antimon$-monotonicty, any model that \emph{violates} the formula will continue to violate it no matter what traces are added.

\begin{definition}\label{def:monotone}
    A \sohyperltlstar{} formula $\varphi$ is $\mon$-monotone (resp.~$\antimon$-monotone) if for any set of traces $\mathbb{T}$ such that $\mathbb{T} \models \varphi$ (resp.~$\mathbb{T} \not\models \varphi$), for any larger set $\mathbb{T}' \supseteq \mathbb{T}$, we have $\mathbb{T}' \models \varphi$ (resp.~$\mathbb{T}' \not\models \varphi$).    
\end{definition}

We say that a formula is monotone if it is either $\mon$-monotone or $\antimon$-monotone. 
Once we detect that a $\mon$-monotone (resp.~$\antimon$-monotone) formula holds (resp.~is violated) on the set of traces seen so far, our monitor can conclude that $\varphi$ holds (resp.~does not hold) regardless of what traces will be observed in the future.
The CK formula in \Cref{ex:common:knowledge:fixpoint} is $\antimon$-monotone, i.e., whenever the current set of traces violates CK, no additional set of traces can change it.
In many cases, monotonicity implies monitorability:

\begin{proposition}
Let $\varphi$ be a \sohyperltlstar{} formula. If one of the following holds, then $\varphi$ is monitorable. 
\begin{enumerate}
    \item $\varphi$ is $\mon$-monotone and has at least one finite model (i.e., a finite set of traces $\mathbb{T}$ such that $\mathbb{T} \models \varphi$). 
    \item $\varphi$ is $\antimon$-monotone and $\neg \varphi$ has at least one finite model. 
\end{enumerate}
\end{proposition}
The proof directly follows from the definitions of monotonicity (\Cref{def:monotone}) and monitorability (\Cref{def:hypermonitorabilitysequential}).
The converse statement does not hold, even for the first-order fragment of \sohyperltlstar.

\begin{example}
    Consider the formula 
    \begin{align*}
    	\varphi := \forall \pi. \exists \pi'. (\pi \neq \pi' \wedge \LTLglobally(a_\pi \leftrightarrow a_{\pi'})) \vee \exists \pi''. \LTLeventually(b_{\pi''})
    \end{align*}
    and the trace sets $\mathbb{T} := \big\{ \{a,c\}^m, \{ a\}^m \big\}$, $\mathbb{T}' := \mathbb{T}\cup \big\{ \{c\}^m \big\}$, and $\mathbb{T}'' := \mathbb{T}' \cup \big\{ \{b\}^m \big\}$.
    It is easy to see that $\mathbb{T}\models \varphi$, $\mathbb{T}' \not\models \varphi$, and $\mathbb{T}'' \models \varphi$. 
    Formula $\varphi$ thus cannot be monotone.
Yet, $\varphi$ is monitorable: for every set $U$, we can add $V :=\big \{ \{b\}^m\big\}$, and every extension of $U\cup V$ satisfies $\varphi$ (cf.~\Cref{def:hypermonitorabilitysequential}).\demo

\end{example}

\paragraph{An Inference System for Monotonicity}

\begin{figure}[!t]
\begin{mathpar}
 \inferrule
  { }
  {\Gamma \vdash a_\pi : \mon}

 \inferrule
  { }
  {\Gamma \vdash a_\pi : \antimon}
 
  \inferrule
  {\Gamma \vdash \varphi : \mon }
  {\Gamma \vdash \neg \varphi : \antimon}
  
  \inferrule
  {\Gamma \vdash \varphi : \antimon }
  {\Gamma \vdash \neg \varphi : \mon} \\

  \inferrule
  {\circ \in \{\X, \LTLprevious\} \\ \Gamma \vdash \varphi : \alpha}
  {\Gamma \vdash \circ \, \varphi : \alpha}

  \inferrule
  { \hspace{6.15em} \\ \Gamma \vdash \varphi_1 : \alpha \\\\
  \circ \in \{\land, \LTLuntil, \LTLsince\}  \\ \Gamma \vdash \varphi_2 : \alpha}
  {\Gamma \vdash \varphi_1 \circ \varphi_2 : \alpha} \\

  \inferrule
  {X \in \Gamma \\ \Gamma \vdash \varphi : \mon}
  {\Gamma \vdash \exists \pi \in X\ldot \varphi : \mon}

  \inferrule
  {X \in \Gamma \\ \Gamma \vdash \varphi : \antimon}
  {\Gamma \vdash \forall \pi \in X\ldot \varphi : \antimon} \\

  \inferrule
  {\Gamma \cup \{X\} \vdash \varphi : \alpha}
  {\Gamma \vdash \mathit{fix}(X, \xi_1, \ldots, \xi_k)\ldot \varphi : \alpha}

  \inferrule
  {\Gamma \vdash \varphi : \alpha}
  {\Gamma \vdash \quant X\ldot \varphi : \alpha}
 \end{mathpar}
	\vspace{-2mm}
     \caption{Inference system for monotonicity}
    \label{fig:type-system}
\end{figure}

In our monitoring algorithm, we use monotonicity to provide definitive monitoring answers. 
To statically determine the monotonicity of (sub)formulas, we use a deductive (type-like) inference system. 
The judgments in our system are of the form $\Gamma \vdash \varphi : \alpha$, where 
$\alpha \in \{ \mon, \antimon \}$, and 
$\Gamma = \{X_1, \ldots, X_n\}$ is a context assumption. 
Intuitively, $X \in \Gamma$ assumes that $X$ 
has a unique model that only grows when more traces arrive (i.e., the model is monotonically increasing); and $\Gamma \vdash \varphi : \mon$ (resp.~$\Gamma \vdash \varphi : \antimon$) implies that, under the context assumptions in $\Gamma$, $\varphi$ is $\mon$-monotone (resp.~$\antimon$-monotone).
The inference rules of our system are depicted in \Cref{fig:type-system}.
Most rules in our system are straightforward: atomic propositions are both $\mon$-monotone and $\antimon$-monotone, temporal operators preserve monotonicity, and negation ``flips'' monotonicity. 
More interesting are the rules for quantification. 
A formula $\varphi = \exists \pi \in X\ldot \varphi'$ is $\mon$-monotone if $\varphi'$ is $\mon$-monotone and 
$X \in \Gamma$ (i.e.,
the model of $X$ only increases):
assume that $\varphi$ holds on the current set of traces $\mathbb{T}$, i.e., 
there exists a trace $t\in X$ that is a witness for the satisfaction of $\varphi'$. 
As $X \in \Gamma$, the model of $X$ can only grow when more traces are added to $\mathbb{T}$, so we can always reuse $t$ as a witness.
Likewise, models of a fixpoint-defined second-order variable $X$ only grow larger with the arrival of new traces, so we add $X$ to $\Gamma$.
For full second-order quantification, we can make no assumption on how model(s) behave when future traces arrive, so we add no assumptions. 
In our inference system, we initially set the context to be $\Gamma = \{\systemvar \}$, as the set of all traces (bound to the special variable $\systemvar \in \sovars$) only grows larger.
An easy induction shows:

\begin{proposition}\label{prop:mon-type-system}
    Assume $\{\systemvar \} \vdash \varphi : \mon$ (resp.~$\{\systemvar \} \vdash \varphi : \antimon$), then $\varphi$ is $\mon$-monotone (resp.~$\antimon$-monotone).
\end{proposition}

Note that our system is not complete, i.e., the converse of \Cref{prop:mon-type-system} does not hold. 
It is easy to see that our system allows us to conclude that the CK property in \Cref{ex:common:knowledge:fixpoint} is $\antimon$-monotone.

\section{Monitoring Algorithm}\label{sec:alg}

\begin{algorithm}[!t]
    \caption{Monitoring algorithm for formula $\varphi$}\label{alg:monitor}
\begin{mycode}
$\mathit{monMap}$ = computeMonotonicity($\varphi$)
$\mathbb{T}$ = $\emptyset$
while ($t$ = getNextTrace()):
@@$\mathbb{T}$ = $\mathbb{T} \cup \{t\}$
@@$\mathit{res}$ = check($\emptyset$,$\emptyset$,$0$,$\mathbb{T}$,$\varphi$)
@@if $\mathit{res}$ = true and $\mon \in \mathit{monMap}$($\varphi$) then return SAT
@@if $\mathit{res}$ = false and $\antimon \in \mathit{monMap}$($\varphi$) then return UNSAT
\end{mycode}
\end{algorithm}

We depict our basic monitor for a formula $\varphi$ in \Cref{alg:monitor}.
As a first step, we (1) compute a monotonicity map $\mathit{monMap}$ using the inference system from \Cref{fig:type-system}. 
Here, $\mathit{monMap}$ maps each \emph{subformula} $\varphi'$ of $\varphi$ to $\mathit{monMap}(\varphi') \in  2^{\{\mon, \antimon\}}$ indicating whether $\varphi'$ -- in a given context $\Pi, \Delta, i$ --  is monotone; and (2) initialize the set $\mathbb{T}$ of system traces to the empty set. 
During the execution of the monitor, 
whenever a new trace $t$ arrives, we add it to $\mathbb{T}$ and check if $\mathbb{T}$
satisfies $\varphi$ by calling function %
\checkfun{} (which we discuss in \Cref{sec:inc}). 
As we argued in \Cref{sec:monotonicity}, our monitor can only provide a definitive \lstinline[style=default,language=custom-lang]|SAT| answer in case $\varphi$ is $\mon$-monotone and the current traces satisfy $\varphi$, and \lstinline[style=default,language=custom-lang]|UNSAT| in case it is $\antimon$-monotone and the current set of traces violates $\varphi$.
Otherwise, we await further traces and repeat.

\subsection{Incremental Model Checking}\label{sec:inc}

\begin{algorithm}[!t]
    \caption{Incremental model-checking}\label{alg:imc}
    
\newsavebox{\sati}
\begin{lrbox}{\sati}
	\begin{tikzpicture}[baseline=-0.6ex]
		\node[fill=black!8,outer sep=0pt,inner xsep=0pt, inner ysep=0pt, rounded corners=2pt, minimum height=\ht\strutbox+1pt] {\lstinline[style=default,language=custom-lang]|if $\mon \in \mathit{monMap}$($\varphi$) and $h_\mathit{sat}$($\Pi$,$\Delta$,$i$,$\varphi$) = true then|};
	\end{tikzpicture}
\end{lrbox}

\newsavebox{\satii}
\begin{lrbox}{\satii}
	\begin{tikzpicture}[baseline=-0.6ex]
		\node[fill=black!8,outer sep=0pt,inner xsep=0pt, inner ysep=0pt, rounded corners=2pt, minimum height=\ht\strutbox+1pt] {\lstinline[style=default,language=custom-lang]|return true|};
	\end{tikzpicture}
\end{lrbox}

\newsavebox{\satiii}
\begin{lrbox}{\satiii}
	\begin{tikzpicture}[baseline=-0.6ex]
		\node[fill=black!8,outer sep=0pt,inner xsep=1pt, inner ysep=0pt, rounded corners=2pt, minimum height=\ht\strutbox+1pt] {\lstinline[style=default,language=custom-lang]|if $\antimon \in \mathit{monMap}$($\varphi$) and $h_\mathit{sat}$($\Pi$,$\Delta$,$i$,$\varphi$) = false then|};
	\end{tikzpicture}
\end{lrbox}

\newsavebox{\satiiii}
\begin{lrbox}{\satiiii}
	\begin{tikzpicture}[baseline=-0.6ex]
		\node[fill=black!8,outer sep=0pt,inner xsep=1pt, inner ysep=0pt, rounded corners=2pt, minimum height=\ht\strutbox+1pt] {\lstinline[style=default,language=custom-lang]|return false|};
	\end{tikzpicture}
\end{lrbox}

\newsavebox{\satiiiii}
\begin{lrbox}{\satiiiii}
	\begin{tikzpicture}[baseline=-0.6ex]
		\node[fill=black!8,outer sep=0pt,inner xsep=1pt, inner ysep=0pt, rounded corners=2pt, minimum height=\ht\strutbox+1pt] {\lstinline[style=default,language=custom-lang]|$h_\mathit{sat}$($\Pi$,$\Delta$,$i$,$\varphi$) = $\mathit{res}$|};
	\end{tikzpicture}
\end{lrbox}

\newsavebox{\opti}
\begin{lrbox}{\opti}
	\begin{tikzpicture}[baseline=-0.6ex]
		\node[fill=black!8,outer sep=0pt,inner xsep=1pt, inner ysep=0pt, rounded corners=2pt, minimum height=\ht\strutbox+1pt] {\lstinline[style=default,language=custom-lang]|if $h_\mathit{wit}$($\Pi, \Delta, i,\varphi$) = $t$ then order($t$,$\Delta(X)$) else|};
	\end{tikzpicture}
\end{lrbox}

\newsavebox{\optii}
\begin{lrbox}{\optii}
	\begin{tikzpicture}[baseline=-0.6ex]
		\node[fill=black!8,outer sep=0pt,inner xsep=1pt, inner ysep=0pt, rounded corners=2pt, minimum height=\ht\strutbox+1pt] {\lstinline[style=default,language=custom-lang]|$h_\mathit{wit}$($\Pi$,$\Delta$,$i$,$\varphi$) = $t$|};
	\end{tikzpicture}
\end{lrbox}

\newsavebox{\fixi}
\begin{lrbox}{\fixi}
	\begin{tikzpicture}[baseline=-0.6ex]
		\node[fill=black!8,outer sep=0pt,inner xsep=1pt, inner ysep=0pt, rounded corners=2pt, minimum height=\ht\strutbox+1pt] {\lstinline[style=default,language=custom-lang]|if $h_\mathit{fix}$($\Pi$,$\Delta$,$i$,$\mathit{fix}(X, \xi_1, \ldots, \xi_k)$) = $A''$ then $A''$ else|};
	\end{tikzpicture}
\end{lrbox}

\newsavebox{\fixii}
\begin{lrbox}{\fixii}
	\begin{tikzpicture}[baseline=-0.6ex]
		\node[fill=black!8,outer sep=0pt,inner xsep=1pt, inner ysep=0pt, rounded corners=2pt, minimum height=\ht\strutbox+1pt] {\lstinline[style=default,language=custom-lang]|$h_\mathit{fix}$($\Pi$,$\Delta$,$i$,$\mathit{fix}(X, \xi_1, \ldots, \xi_k))$ = $A$|};
	\end{tikzpicture}
\end{lrbox}
   
\begin{mycode}
def check($\Pi$,$\Delta$,$i$,$\mathbb{T}$,$\varphi$):
@@(*\usebox{\sati}*)(*\label{line:imc-hash1}*)
@@@@(*\usebox{\satii}*)
@@(*\usebox{\satiii}*) (*\label{line:imc-hash2}*)
@@@@(*\usebox{\satiiii}*)
@@$\mathit{res}$ = match $\varphi$ with (*\label{line:imc-match}*)
@@@@| $a_\pi$: if $a \in \Pi(\pi)(i)$ then return true else return false (*\label{line:imc-ap}*)
@@@@| $\neg \varphi'$: return (not (check($\Pi$,$\Delta$,$i$,$\mathbb{T}$,$\varphi'$))) (*\label{line:imc-neg}*)
@@@@| $\exists \pi \in X \ldot \varphi'$:(*\label{line:imc-fo}*) 
@@@@@@$A$ = (*\usebox{\opti}*) $\Delta(X)$ (*\label{line:imc-wit-hash-look}*)
@@@@@@for $t$ in $A$: (*\label{line:imc-fo-iter}*)
@@@@@@@@if check($\Pi[\pi \mapsto t]$,$\Delta$,$i$,$\mathbb{T}$,$\varphi'$) then (*\label{line:imc-fo-check}*)
@@@@@@@@@@(*\usebox{\optii}*) (*\label{line:imc-wit-hash-store}*)
@@@@@@@@@@return true
@@@@@@return false
@@@@| $\exists X \ldot \varphi'$:(*\label{line:imc-so}*) 
@@@@@@for $A \subseteq \traceSet$:
@@@@@@@@if check($\Pi$,$\Delta[X \mapsto A]$,$i$,$\mathbb{T}$,$\varphi'$) then
@@@@@@@@@@return true 
@@@@@@return false
@@@@| $\mathit{fix}(X, \xi_1, \ldots, \xi_k) \ldot \varphi'$:(*\label{line:imc-fix}*) 
@@@@@@$A'$ = (*\usebox{\fixi}*) $\emptyset$ (*\label{line:imc-fix-hash-look}*)
@@@@@@$A$ = computeFix($\Pi$,$\Delta$,$i$,$\mathbb{T}$,$\mathit{fix}(X, \xi_1, \ldots, \xi_k)$,$A'$)(*\label{line:imc-fix-comp}*)
@@@@@@(*\usebox{\fixii}*) (*\label{line:imc-fix-hash-store}*)
@@@@@@return check($\Pi$,$\Delta[X \mapsto A]$,$i$,$\mathbb{T}$,$\varphi'$)(*\label{line:imc-fix-check}*)
@@(*\usebox{\satiiiii}*) (*\label{line:imc-hash-store}*)
@@return $\mathit{res}$
\end{mycode}
\end{algorithm}

The core of our monitoring algorithm lies in our recursive model-checking function %
\checkfun{}
given in \Cref{alg:imc}. 
On a high level, %
\checkfun{} casts the semantics of \sohyperltlstar{} into an executable program.
For now, we ignore the program fragments marked with a gray background; these concern the hashing-based optimizations we will discuss in \Cref{sec:opt}.
In line \ref{line:imc-match}, we match on the structure of $\varphi$. 
We only include a selection of cases (the others can be inferred easily).
In case $\varphi$ is an AP, 
we can immediately decide its truth value (line \ref{line:imc-ap}); for the cases of negation, conjunction, and temporal operators we perform the expected recursive call(s).
For first-order quantification (line~\ref{line:imc-fo}), we iteratively check all traces assigned by $\Delta$. 
For full second-order quantification (line \ref{line:imc-so}), we check all subsets of $\traceSet$ . 
In the case of fixpoint-based second-order quantification (line \ref{line:imc-fix}), we compute the (unique) solution to the fixpoint definition (line \ref{line:imc-fix-comp}) by calling \lstinline[style=default,language=custom-lang]|computeFix| (discussed in \Cref{sec:sub:fixpoint-compo}).

\subsection{Fixpoint Computation}\label{sec:sub:fixpoint-compo}

As we argue in \Cref{sec:sohyperltlstar}, fixpoints are expressive enough for most practical properties and are much easier to handle algorithmically. 
Given a set $\mathbb{T}$ with $n$ traces, 
we can compute the fixpoint in polynomial time solution using Knaster-Tarski fixpoint iteration~\cite{tarski1955lattice} instead of trying all possible $2^n$ subsets of $\mathbb{T}$ (as needed for full second-order quantification). 
The input to \Cref{alg:fixpoint} is the fixpoint formula and the current set $A$ of traces in the (to-be) fixpoint set. 
Initially, $A = \emptyset$. 
We check if $A$ satisfies the fixpoint constraints $\xi_1, \ldots, \xi_k$: if we find traces $t_1 \in \Delta(X_1), \ldots, t_n \in \Delta(X_n)$ that satisfy the step constraint  (which we check via a mutually recursive call in line \ref{line:fix-check}), we add the trace required by the fixpoint constraint to $A$ and repeat via a recursive call (line~\ref{line:fix-rec-call}). 

\subsection{Monitoring Optimizations}\label{sec:opt}
When implemented without optimizations, \Cref{alg:imc} checks the \sohyperltlstar{} formula on the current set of traces, re-performing this check whenever a new trace arrives. 
This is often infeasible in practice, as with each new trace, the number of traces and, therefore, the verification time increases.
Instead, whenever we add a trace $t$
to the current set of traces $\mathbb{T}$, we want to reuse as much of the computation we have already performed on $\mathbb{T}$ in previous steps.
We discuss three areas where such re-usage is possible: hashing of verification results for subformulas, hashing of fixpoint results, and hashing of witnesses. 
\begin{algorithm}[!t]
    \caption{Fixpoint computation}\label{alg:fixpoint}

\begin{mycode}
def computeFix($\Pi$,$\Delta$,$i$,$\mathbb{T}$,$\mathit{fix}(X, \xi_1, \ldots, \xi_k)$,$A$):
@@$\Delta'$ = $\Delta[X \mapsto A]$ (*\label{line:fix-init}*)
@@for $\forall \pi_1 \in X_1 \ldots \forall \pi_n \in X_n\ldot \varphi_\mathit{step} \to \pi \in X $ in $\{\xi_1, \ldots, \xi_k\}$:
@@@@for $t_1 \in \Delta'(X_1), \ldots, t_n \in \Delta'(X_n)$: 
@@@@@@$\Pi'$ = $\Pi[\pi_1 \mapsto t_1, \ldots, \pi_n \mapsto t_n]$
@@@@@@if $\Pi'(\pi) \not\in A$ and check($\Pi'$,$\Delta'$,$i$,$\mathbb{T}$,$\varphi_\mathit{step}$) then  (*\label{line:fix-check}*)
@@@@@@@@return computeFix($\Pi$,$\Delta$,$i$,$\mathbb{T}$,$\mathit{fix}(X, \xi_1, \ldots, \xi_k)$,$A \cup \{\Pi'(\pi) \}$) (*\label{line:fix-rec-call}*)
@@return $A$
\end{mycode}
\end{algorithm}

\paragraph{SAT Hashing}
Our first optimization is the hashing of verification results of the \checkfun{} for \emph{subformulas}. 
If -- for a given evaluation context $(\Pi, \Delta, i)$ -- a $\mon$-monotone (resp.~$\antimon$-monotone) subformula $\varphi$ was satisfied (resp.~violated) in a previous iteration, then $\varphi$ remains satisfied (resp.~violated) in context $(\Pi, \Delta, i)$, even if more traces are added to $\traceSet$.
We therefore use a (hashing) function $h_\mathit{sat}$ (initially set to the function with empty domain) and check if we have hashed $(\Pi, \Delta, i, \varphi)$ (\Cref{alg:imc}, lines \ref{line:imc-hash1} and \ref{line:imc-hash2}).
Here, we, e.g., write \lstinline[style=default,language=custom-lang]|$h_\mathit{sat}$($\Pi$,$\Delta$,$i$,$\varphi$) = true| to denote that $(\Pi, \Delta, i, \varphi)$ is contained in the domain of $h_\mathit{sat}$ and maps to \lstinline[style=default,language=custom-lang]|true|.
If this hashing does not suffice, we evaluate $\varphi$ and update $h_\mathit{sat}$ in line \ref{line:imc-hash-store}.

\paragraph{Fixpoint Hashing}
The second optimization we use is hashing fixpoint solutions. 
In any given content $(\Pi, \Delta, i)$, the fixpoint solution only increases when new traces are added to $\mathbb{T}$.
When we compute a fixpoint solution (\Cref{alg:imc}, line \ref{line:imc-fix-comp}), we thus do not need to restart the computation from scratch. 
Instead, we hash the solution from the previous iteration using function $h_\mathit{fix}$. 
If we hashed $(\Pi, \Delta, i, \mathit{fix}(X, \xi_1, \ldots, \xi_k))$, we can this previous solution as the starting point for the current fixpoint computation. 
Otherwise, we start the computation from the empty set (line~\ref{line:imc-fix-hash-look}).
After calling \lstinline[style=default,language=custom-lang]|computeFix|, we update $h_\mathit{fix}$ and map $(\Pi, \Delta, i, \mathit{fix}(X, \xi_1, \ldots, \xi_k))$ to the newly computed solution $A$ (line \ref{line:imc-fix-hash-store}).

\paragraph{Witness Hashing}
Our last hashing-based optimization concerns the analysis of first-order quantification. 
In line~\ref{line:imc-fo-iter}, we need to iterate over all traces in the current model. 
While we cannot avoid this in the worst case, we can optimize the order in which we explore the traces. 
Concretely, whenever we find a witness trace $t$, we store it using the $h_\mathit{wit}$ hashing function (line~\ref{line:imc-wit-hash-store}). 
In the next iteration, we attempt to reorder $\Delta(X)$ such that the witness from the previous iteration will be explored first (line \ref{line:imc-wit-hash-look}). 
Intuitively, if $t$ is a witness for the satisfaction of some formula $\exists \pi. \varphi'$, then
$t$ remains a promising witness when more traces are added to $\mathbb{T}$; exploring $t$ early can thus save significant time. 

\paragraph{Prefix and Postfix Trees}
Our last optimization is applicable to all hashing techniques above.
Instead of storing the set $\mathbb{T}$ of all traces explicitly as a list, 
we store it as a tree, i.e., we store all traces based on their common prefixes (resp.~postfixes). This has two advantages:
It reduces memory overhead, particularly in cases where many traces share parts; 
and, secondly, a tree-based representation allows for even more efficient hashing. 
For example, if at time step $i$, we compute a fixpoint solution based only on future temporal operators (i.e., no past modalities), 
then all traces that, starting from step $i$, share the same postfix can be treated as equal during the fixpoint computation.
We can thus lift the computation and operate on nodes in the postfix tree rather than on concrete traces. 

\section{Experiments}\label{sec:experiments}
We have implemented our monitoring algorithm and the optimizations from \Cref{sec:opt} in a tool called \moso{}.
In this section, we demonstrate that \moso{} can monitor complex second-order hyperproperties that are out of reach of existing monitoring and model-checking frameworks.

\subsection{Running Example}
The first experiment is the scalability of \moso{} in the length of the traces.
We verify the formula in~\Cref{ex:common:knowledge:fixpoint} against traces of the running example in~\Cref{fig:sender:receiver}.
We generate traces for each instance of lengths 20, 30, 40, ..., 80 and measure the time it takes to find the violation for the starting trace $s^nr^n$ in \Cref{tab:running}.
The monitor correctly concludes in all runs that the common knowledge formula is violated.

\begin{table}[!t]
    \caption{We monitor CK in \Cref{ex:common:knowledge:fixpoint} for varying trace lengths and trace numbers (\textbf{\# traces}). 
    We report \moso{}'s average runtime in seconds ($\boldsymbol{t}$). }\label{tab:running}
    \centering
    \small
    \begin{tabular}{lccccccc}
        \toprule
        \textbf{length} & 20 & 30 & 40 & 50 & 60 & 70 & 80 \\
        \textbf{\# traces} & 35 & 55 & 75 & 95 & 115 & 135 & 155\\
        \midrule
        $\boldsymbol{t}$ & 0.51 & 1.51 & 5.98 & 18.40 & 48.83 & 111.52 & 230.71 \\
        \bottomrule
    \end{tabular}
\end{table}

\subsection{Common Knowledge in MASs}
We monitor common knowledge in a scalable instance of the muddy children's puzzle~\cite{DBLP:books/mit/FHMV1995}.
The muddy children puzzle is a MAS between children $1, \ldots, n$, where a number of rounds of communication are used to achieve common knowledge about how many children are muddy.
For each child (agent) $i$, we have an AP $m_i$ indicating if $i$ is muddy and an AP $d_i$ indicating if $i$ declared that it is muddy. 
We can then express that the muddiness of all children is CK after $b \in \nat$ steps of communication (which we call the \emph{communication-bound}) as the following \sohyperltlstar{} formula $\varphi_{\mathit{mud}}$:
\begin{align*}
    \forall \pi& \in \systemvar\ldot \LTLnext^b \mathit{fix}\big(X, \xi_1, \xi_2\big)\ldot \forall \pi_1 \in X\ldot \forall \pi_2 \in X\ldot \!\!\!\!\!\! \bigwedge_{a\in \{m_1,\ldots,m_n\}} \!\!\!\!\!\! \big( {a}_{\pi_1} \leftrightarrow {a}_{\pi_2} \big),
\end{align*}
where $\xi_1 := \mathit{true} \to \pi \in X$, and $\xi_2$ is defined as
\begin{align*}
    \forall \pi_1 \in X\ldot \forall \pi_2 \in \systemvar\ldot\big(\!\!\!\!\!\bigvee_{i\in\{1,\ldots,n\}} \LTLpastglobally ( \!\!\! \bigwedge_{\substack{a \in \{d_1, \ldots, d_n\} \cup \\ \{m_1, \ldots, m_{i-1}\}\cup \\ \{ m_{i+1}, \ldots, m_n \}}} \!\!\!\!\!\!\!\! a_{\pi_1} \leftrightarrow a_{\pi_2}) \big) \rightarrow \pi_2 \in X.
\end{align*}
Intuitively, the fixpoint captures all traces that some child $i$ cannot distinguish from $\pi$ in the first $b$ steps, i.e., we add all traces that agree on all APs except $m_i$.
Formula $\varphi_{\mathit{mud}}$ then requires that all traces in $X$ agree on the muddiness of all children.
The muddy children's MAS violates $\varphi_{\mathit{mud}}$ if and only if $b < n$.

\begin{table}[!t]
    \caption{We monitor CK in the muddy children's game. We report \moso{}'s average running time in seconds ($\boldsymbol{t}$). We also depict the runtime on the non-fixpoint-based formulation of CK ($\boldsymbol{t}_\mathit{noFix}$).
    The timeout (TO) is set to 1 min.}\label{tab:muddy-time}
    \centering
    \small
    \begin{tabular}{lccccccccc}
        \toprule
        \textbf{\# children} & 2 & 3 & 4 & 5 & 6 & 7 & 8 & 9   \\
        \midrule
        $\boldsymbol{t}$ & 0.1 & 0.2 & 0.2 & 0.2 & 0.3 & 0.4 & 0.4 & 0.4 \\
        $\boldsymbol{t}_\mathit{noFix}$ & 0.1 & 1.1 & TO & TO & TO & TO & TO & TO \\
        \bottomrule
    \end{tabular}
\end{table}

\paragraph{Scalability}
For the game with $n$ children, we sample random traces and use \moso{} to monitor CK (we set the communication-bound to be $b = \lceil \frac{n}{2} \rceil$ so CK does not hold). 
We report the runtime in \Cref{tab:muddy-time}.
We observe that even if the number of children grows, our monitor will find violations to common knowledge quickly and thus runs very effectively. 
In contrast, \emph{model-checking} of CK \cite{BeutnerFFM23} is, currently, only possible up to $n = 4$.
This supports our claim that monitoring second-order hyperproperties is a useful lightweight technique that scales in settings where full verification is infeasible.

\paragraph{Fixpoints vs Second-Order Sets}
The muddy children's puzzle also highlights the importance of fixpoint formulas compared to arbitrary second-order quantification. 
To test this empirically, we consider the same CK property but express it directly using full second-order quantification (similar to the formula in \Cref{sec:intro}). 
We give the runtime of \moso{} on the non-fixpoint-based formula in \Cref{tab:muddy-time} ($\boldsymbol{t}_\mathit{noFix}$).
When using general second-order quantification, \moso{} already times out for 4 children (due to the exponential cost of considering all subsets of traces).
This attests to the importance of fixpoints for scalable monitoring.

\paragraph{Impact of Communication-bound.}

\begin{table}[!t]
    \caption{We depict the (average) number of traces the monitor processes before concluding a violation of CK. 
    }\label{tab:muddy-steps-to-violation}
    
    \centering
    \small
    \begin{tabular}{|c|c|c|c|c|c|c|c|}
        \cline{3-8}
        \multicolumn{2}{c|}{}& \multicolumn{6}{c|}{\textbf{communication-bound}} \\
        \cline{3-8}
        \multicolumn{2}{c|}{} & \textbf{0} & \textbf{1} & \textbf{2} & \textbf{3} & \textbf{4} & \textbf{5} \\
        \hline
        \parbox[t]{2mm}{\multirow{5}{*}{\rotatebox[origin=c]{90}{\textbf{\# children}}}}& \textbf{2} & 1.3 & 7.0 & - & - & - & - \\
        \cline{2-8}
        & \textbf{3} & 2.1 & 4.5 & 15.0 & - & - & - \\
        \cline{2-8}
        & \textbf{4} & 2.1 & 4.2 & 5.3 & 31.0 & - & - \\
        \cline{2-8}
        & \textbf{5} & 2.9 & 3.9 & 5.7 & 33.1 & 63.0 & - \\
        \cline{2-8}
        & \textbf{6} & 4.3 & 4.3 & 8.1 & 12.7 & 43.3 & 127.0 \\
        \hline
    \end{tabular}
\end{table}

Finding violations to the CK formula heavily depends on the communication-bound.
With increasing communication-bound, the probability of sampling traces that violate CK decreases as more information is observed by each child.
On average, we thus need to see more traces until our monitor reports a violation. 
We test this empirically by monitoring the muddy children puzzle with a varying number of children and communication-bound and report on the average number of traces (across 10 monitor runs) the monitor observes before concluding that CK does not hold. 
We depict the results in \Cref{tab:muddy-steps-to-violation}.
A larger communication bound clearly increases the number of traces the monitor processes until a violation is reported.

\subsection{Impact of Optimizations}

Common knowledge also serves as a useful baseline to highlight the importance of our hashing-based optimizations. 
To demonstrate this, we check the following CK formula
\begin{align*}
    \exists \pi \in{} &\systemvar\ldot \LTLeventually \mathit{fix}\Big(X, \mathit{true} \to \pi \in X, \forall \pi_1 \in X\ldot \forall \pi_2 \in \systemvar\ldot\\ 
    &\!\big(\LTLpastglobally (a_{\pi_1} \leftrightarrow a_{\pi_2}) \lor \LTLpastglobally (b_{\pi_1} \leftrightarrow b_{\pi_2})\big) \rightarrow \pi_2 \in X \Big)\ldot \forall \pi' \in X\ldot c_{\pi'},
\end{align*}
which states that on \emph{some} path $\pi$, $c$ is eventually CK, given that one agent observes AP $a$ and one agent observes AP $b$.
We use \moso{} to check the above formula on randomly generated instances with a varying number of traces.
We depict the results in \Cref{fig:opt}.
We observe that all optimizations improve upon the baseline. 

\begin{figure}[!t]
    \centering
    \includegraphics[width=0.95\linewidth]{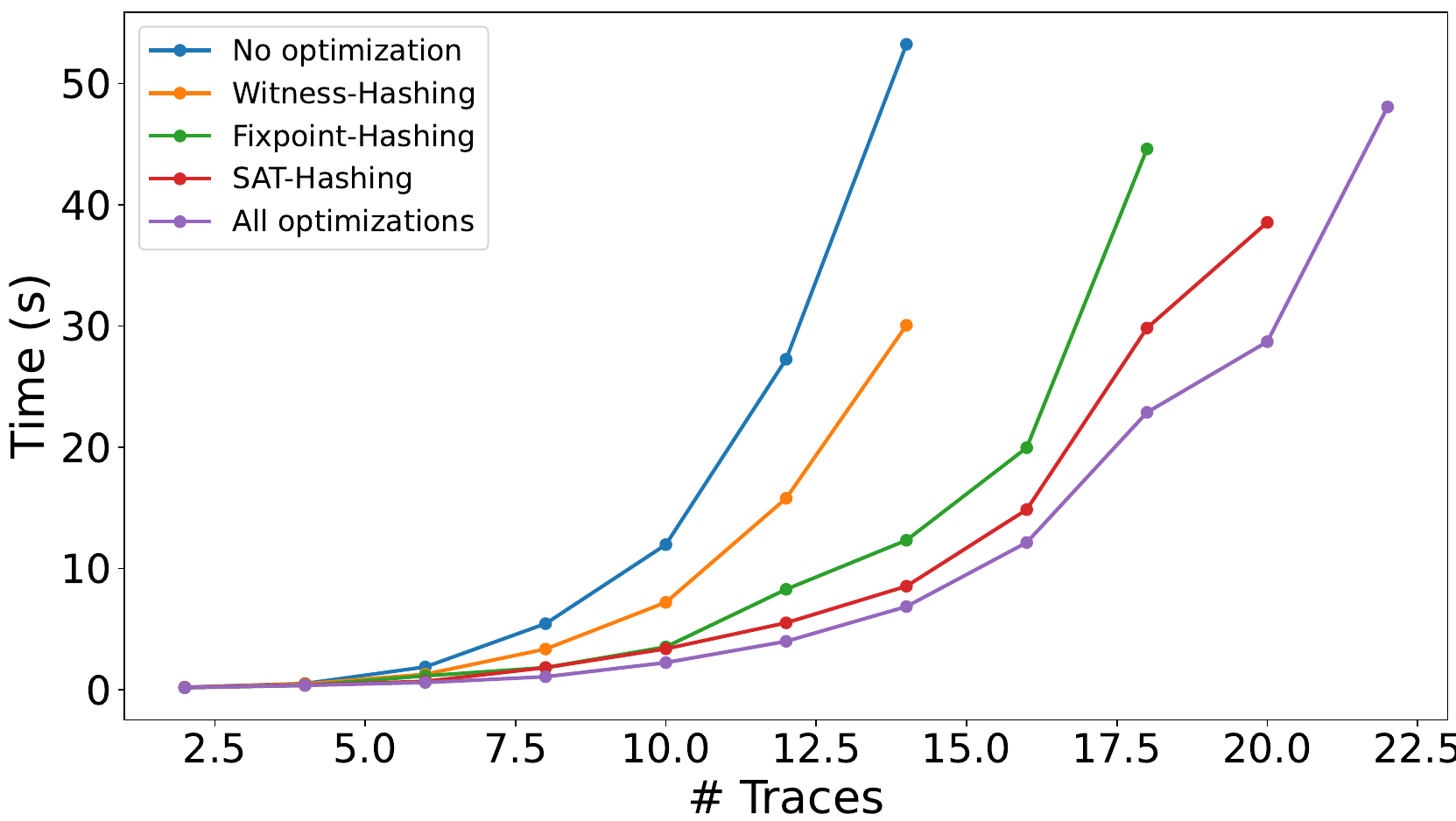}
    \vspace{-2mm}
    \caption{We compare the optimizations implemented in \moso{} as a cactus plot.
    The timeout is set to 60 sec.}\label{fig:opt}
\end{figure}

\begin{table}[!t]
    \caption{We monitor the existence of a path from source to target based on random paths in random graphs.
    For each graph size, we sample 1000 random traces and report the average number of traces that \moso{} processes to reach a verdict (\textbf{\# traces}) and the average runtime in seconds ($\boldsymbol{t}$).
    }\label{tab:planning}
    \small
    \setlength\tabcolsep{4.5pt}

    \begin{tabular}{clccccccc}
        \toprule
        & \textbf{size} & 55 & 70 & 85 & 100 & 115 & 130\\
        \midrule
        \multirow{2}{*}{\textbf{t-sen.}} & $\boldsymbol{t}$  & 1.66 & 4.08 & 5.44 & 38.1 & 35.52 & 88.19 \\
          & \textbf{\# traces}  & 1.03 & 1.76 & 3.2 & 3.91 & 4.86 & 5.04 \\
      \midrule
        \multirow{2}{*}{\textbf{t-ins.}} & $\boldsymbol{t}$ & 45.2 & 61.45 & 72.15 &  112.15 & 76.5 & 137.95 \\ 
         & \textbf{\# traces} & 52.15 & 77.35 & 98.25 & 83.95 & 125.75 & 104.55\\
     \bottomrule
    \end{tabular}
    
\end{table}

\subsection{Planning Analysis}
In our last experiment, we use our monitor to detect reachability in a graph based on the observation of random paths. 
For example, when observing traces $t_1 = \mathsf{a} \mathsf{b} \mathsf{c}$, and 
$t_2 = \mathsf{b}\mathsf{b} \mathsf{d}$, we can conclude that the trace $t = \mathsf{a} \mathsf{b} \mathsf{d}$ exists in the graphs, even without observing it. 
The inference of $t$ can easily be stated as a fixpoint constraint:
We differentiate between time-sensitive (\textbf{t-sen.}) and time-insensitive (\textbf{t-ins.}) paths. 
In the former model, we can only ``combine'' two paths if they visit the same state \emph{at the same time (step)}. 
In the latter model, we can even combine two paths if they visit the same state at \emph{potentially different times}. 
We report on the time and number of observed traces it takes to find a correct plan in \Cref{tab:planning}.
As expected, the time-insensitive case terminates earlier than the time-sensitive case since more combinations are possible, increasing the chance of concluding the existence of a path from source to target.

\section{Conclusion}
We have presented \sohyperltlstar{}, the first logic to express second-order hyperproperties for finite trace models. 
Our experiments show that our monitoring algorithm scales to larger instances well beyond the reach of complete methods such as model checking.  
Monitoring of complex and important second-order hyperproperties (particularly in the MAS domain), is thus a viable middle ground to obtain rigorous guarantees while maintaining scalability. 

\begin{acks}
This work was supported by the European Research Council (ERC) Grant HYPER (101055412), by the German Research Foundation (DFG) as part of TRR 248 (389792660), and by the German Israeli Foundation (GIF) Grant No. I-1513-407.2019.
\end{acks}

\bibliographystyle{ACM-Reference-Format} 
\balance
\bibliography{bib}

\end{document}